\RequirePackage{fix-cm}
\documentclass[smallextended]{svjour3}       
\smartqed  

\usepackage{pslatex}
\usepackage{amsfonts,color,morefloats}
\usepackage{amssymb,amsmath,latexsym}

\newcommand{\wt}{{\mathrm{wt}}}

\newcommand{\ord}{{\mathrm{ord}}}
\newcommand{\Z}{\mathbb{{Z}}}

\newcommand{\tr}{{\mathrm{Tr}}}
\newcommand{\gf}{{\mathrm{GF}}}

\newcommand{\C}{{\mathcal{C}}}

\newcommand{\bD}{{\mathbf{D}}}

\newcommand{\bc}{{\mathbf{c}}}

\usepackage{graphicx}
\begin{document}

\title{It is indeed a fundamental construction of all linear codes
}


\author{Can Xiang       
}


\institute{  C. Xiang \at
                College of Mathematics and Informatics, South China Agricultural University, Guangzhou, 510642, China \\
              \email{cxiangcxiang@hotmail.com}
}

\date{Received: date / Accepted: date}

\maketitle

\begin{abstract}
Linear codes are widely employed in communication systems, consumer electronics, and storage devices. All linear codes over finite fields can be generated by a generator matrix. Due to this, the generator matrix
approach is called a fundamental construction of linear codes. This is the only known construction method
that can produce all linear codes over finite fields. Recently, a defining-set construction of linear codes over finite fields has attracted a lot of attention, and have been employed to produce a huge number of
classes of linear codes over finite fields. It was claimed that this approach can also generate all linear
codes over finite fields. But so far, no proof of this claim is given in the literature. The objective of
this paper is to prove this claim, and confirm that the defining-set approach is indeed a fundamental
approach to constructing all linear codes over finite fields. As a byproduct, a trace representation of all
linear codes over finite fields is presented.

\end{abstract}

\section{Introduction}\label{sec-intro}

Throughout this paper, let $q$ be a power of a prime $p$. An $[n,k,d]$ linear code $\C$ over $\gf(q)$ is a $k$-dimensional subspace of $\gf(q)^n$ with minimum Hamming distance $d$.
Let $A_i$ denote the number of codewords with Hamming weight $i$ in a linear code
$\C$ of length $n$. The {\em weight enumerator} of $\C$ is defined by
$$
1+A_1z+A_2z^2+ \cdots + A_nz^n.
$$
The {\em weight distribution} of $\C$ is the sequence $(1,A_1,\ldots,A_n)$.

An $[n,k]$ linear code $\C$ over $\gf(q)$ is called {\em cyclic} if
$(c_0,c_1, \cdots, c_{n-1}) \in \C$ implies $(c_{n-1}, c_0, c_1, \cdots, c_{n-2})$
$\in \C$.
We can identify a vector $(c_0,c_1, \cdots, c_{n-1}) \in \gf(q)^n$
with
$$
c_0+c_1x+c_2x^2+ \cdots + c_{n-1}x^{n-1} \in \gf(q)[x]/(x^n-1).
$$
In this way, a code $\C$ of length $n$ over $\gf(q)$ corresponds to a subset of the quotient ring
$\gf(q)[x]/(x^n-1)$.
A linear code $\C$ is cyclic if and only if the corresponding subset in $\gf(q)[x]/(x^n-1)$
is an ideal of the ring $\gf(q)[x]/(x^n-1)$.

It is well-known that every ideal of $\gf(q)[x]/(x^n-1)$ is principal. Let $\C=\langle g(x) \rangle$ be a
cyclic code, where $g(x)$ is monic and has the smallest degree among all the
generators of $\C$. Then $g(x)$ is unique and called the {\em generator polynomial,}
and $h(x)=(x^n-1)/g(x)$ is referred to as the {\em check polynomial} of $\C$.

Cyclic codes over finite fields can be generated by a generator matrix, or a generator polynomial,
or a generating idempotent. Under certain conditions, cyclic codes over finite fields have a trace
representation described in the following theorem whose proof is based on Delsarte's Theorem
\cite{Wolf,Dels}.

\begin{theorem}[Wolfmann]\label{thm-tracecycliccodes}
Let $\C$ be a cyclic code of length $n$ over $\gf(q)$ with parity-check polynomial
$h(x)$, where $\gcd(n,q)=1$. Let $\beta$ be a primitive $n$-th root of unity over $\gf(q^m)$, where $m:=\ord_n(q)$ is the
order of $q$ modulo $n$. Let $J$ be a subset of $\Z_n=\{0,1,2, \cdots, n-1\}$ such that
$$
h^*(x)=\prod_{j \in J} m_{\beta^j}(x),
$$
where $m_{\beta^j}(x)$ denotes the minimal polynomial of $\beta^j$ over $\gf(q)$, and $h^*(x)$ is the
reciprocal of $h(x)$. Then $\C$ consists of all the following codewords
$$
c_a(x)=\sum_{i=0}^{n-1} \tr(f_a(\beta^i)) x^i,
$$
where $\tr$ denotes the trace function from $\gf(q^m)$ to $\gf(q)$, and
$$
f_a(x)=\sum_{i \in J} a_j x^j, \ \ a_j \in \gf(q^m).
$$
\end{theorem}

This trace representation of certain sub-classes of cyclic codes was known for a long time. For example, the trace
representation of irreducible dates back at least to Baumert and McEliece \cite{BM72}. A trace
description of irreducible quasi-cyclic codes was ginen in \cite{SD90}. The trace representation
of cyclic codes over finite fields in Theorem \ref{thm-tracecycliccodes} was presented and proved by Wolfmann in
\cite{Wolf} under the restriction that $\gcd(q, n)=1$. It demonstrates another way of generating many cyclic codes
over finite fields.
The importance of this trace representation is mostly demonstrated by its application in determining the weight
distribution (also called the weight enumerator) of cyclic codes over finite fields. The trace representation allows
one to determine the weight distribution of a cyclic code by evaluating certain types of character sums over finite
fields, and has led to a lot of recent progress on the weight distribution problem of cyclic codes \cite{DLLZ,DY13}.

However, to the best knowledge of the author, no similar trace representation of linear codes over finite fields is
available in the literature. Even the trace representation in Theorem \ref{thm-tracecycliccodes} has the restriction
that $\gcd(n, q)=1$ and thus applies to only a special type of cyclic codes. One objective of this paper
is to give a
trace representation of all linear codes over finite fields.

It is well known that all linear codes over finite fields can be generated with a generator matrix.
Because of this fact, the generator matrix approach is a fundamental approach to constructing all
linear codes over finite fields and is the only one. Recently, a defining-set construction of linear
codes over finite
fields has been intensively investigated, and shown to be a promising approach, as many classes of
linear codes with good parameters have been produced. It was claimed in \cite{Dingconf} that all
linear codes over finite fields can be produced with this approach. But so far, no proof has been
seen in the literature. Another objective of this paper is to provide this claim, and confirm that
it is indeed a fundamental construction of all linear codes over finite fields.

\section{A generic construction of linear codes over finite fields}

Throughout this section, let $q$ be a prime power and let $r=q^m$, where $m$ is a positive integer.
Let $\tr$ denote the trace function from $\gf(r)$ to $\gf(q)$ unless otherwise stated.

\subsection{The description of the construction}

Let $D=\{d_1, \,d_2, \,\ldots, \,d_n\} \subseteq \gf(r)$.
We define a code of
length $n$ over $\gf(q)$ by
\begin{eqnarray}\label{eqn-maincode}
\C_{D}=\{(\tr(xd_1), \tr(xd_2), \ldots, \tr(xd_n)): x \in \gf(r)\},
\end{eqnarray}
and call $D$ the \emph{defining set} of this code $\C_{D}$.
Since the trace function is linear, the code $\C_D$ is linear.  By definition, the
dimension of the code $\C_D$ is at most $m$.

Different orderings of the elements of $D$ give different linear codes, which are however
permutation equivalent. Hence, in this paper, we do not distinguish these codes obtained
by different orderings, and do not consider the ordering of the elements in $D$. It should
be noticed that the defining set $D$ could be a multiset, i.e., some elements in $D$ may be
the same.

\subsection{The generator matrix of the trace code $\C_D$}\label{sec-generatormatrix}

Every linear code over a finite field must have a generator matrix. In this subsection, we
derive a generator matrix for the trace code $\C_D$, where
\begin{eqnarray}\label{eqn-definingsetD}
D=\{d_1, \,d_2, \,\ldots, \,d_n\} \subseteq \gf(r).
\end{eqnarray}

Let $\{\alpha_1, \alpha_2, \cdots, \alpha_m \}$ be a basis of $\gf(r)$ over $\gf(q)$,
and $\{\beta_1, \beta_2, \cdots, \beta_m \}$ be its dual basis. By the definition of the
dual basis,
\begin{eqnarray}\label{eqn-dualbasis}
\tr(\alpha_i \beta_j)=\left\{
\begin{array}{ll}
0 & \mbox{ for } i \ne j, \\
1 & \mbox{ for } i = j.
\end{array}
\right.
\end{eqnarray}
Note that every basis of $\gf(q^m)$ over $\gf(q)$ has its dual basis \cite[p. 58]{LN}.

Let
$$
d_i=\sum_{j=1}^m d_{j,i} \alpha_j, \ \mbox{ where all }  d_{j,i} \in \gf(q)
$$
and
$$
x=\sum_{h=1}^m x_h \beta_h, \ \mbox{ where all }  x_h \in \gf(q).
$$
Then we have
\begin{eqnarray}
\tr(d_ix) = \tr\left(\sum_{h=1}^m \sum_{j=1}^m x_{h} d_{j,i} \tr(\alpha_j \beta_h)\right)
= \sum_{h=1}^m x_h d_{h,i}.
\end{eqnarray}
Consequently,
\begin{eqnarray}
\bc_x= (\tr(d_1x), \tr(d_2x), \cdots, \tr(d_nx))=(x_1, x_2, \cdots, x_n) \bD
\end{eqnarray}
where
\begin{eqnarray}\label{eqn-generatormatrixD}
\bD=\left[
\begin{array}{llll}
d_{1,1}  & d_{1,2} & \cdots &  d_{1,n} \\
d_{2,1}  & d_{2,2} & \cdots & d_{2,n} \\
\vdots  & \vdots & \vdots & \vdots \\
d_{m,1}  & d_{m,2} & \cdots & d_{m,n}
\end{array}
\right].
\end{eqnarray}
As a result, $\bD$ is a generator matrix of the code $\C_D$, and depends on the choice of
the basis $\{\alpha_1, \alpha_2, \cdots, \alpha_m \}$. We will need this generator matrix
later in this paper.

\subsection{Weights in the codes $\C_D$}

Define for each $x \in \gf(r)$,
\begin{eqnarray}\label{eqn-mcodeword}
\bc_{x}=(\tr(xd_1), \,\tr(xd_2), \,\ldots, \,\tr(xd_n)),
\end{eqnarray}
The Hamming weight $\wt(\bc_x)$ of $\bc_x$ is $n-N_x(0)$, where
$$
N_x(0)=\left|\{1 \le i \le n: \tr(xd_i)=0\}\right|
$$
for each $x \in \gf(r)$.

It is easily seen that for any $D=\{d_1,\,d_2,\,\ldots, \,d_n\} \subseteq \gf(r)$
we have
\begin{eqnarray}\label{eqn-hn3}
qN_x(0)
&=& \sum_{i=1}^n \sum_{y \in \gf(q)} \tilde{\chi}_1 (y\tr(xd_i)) \nonumber \\
&=& \sum_{i=1}^n \sum_{y \in \gf(q)} \chi_1(yxd_i) \nonumber  \\
&=& n + \sum_{i=1}^n \sum_{y \in \gf(q)^*} \chi_1(yxd_i)  \nonumber \\
&=& n + \sum_{y \in \gf(q)^*} \chi_1(yxD),
\end{eqnarray}
where $\chi_1$ and $\tilde{\chi}_1$ are the canonical additive characters of $\gf(r)$ and $\gf(q)$, respectively,
$aD$ denotes the set
$\{ad: d \in D\}$, and $\chi_1(S):=\sum_{x \in S} \chi_1(x)$ for any subset $S$ of $\gf(r)$.
Hence,
\begin{eqnarray}\label{eqn-weight}
\wt(\bc_x)=n-N_x(0)=\frac{(q-1)n-\sum_{y \in \gf(q)^*} \chi_1(yxD)}{q}.
\end{eqnarray}
Thus, the computation of the weight distribution of the code $\C_D$ reduces to the determination of the value distribution
of the character sum
$$
\sum_{y \in \gf(q)^*} \sum_{i=1}^n \chi_1(yxd_i).
$$

\subsection{Comments on the trace construction and the objective of this paper}

This construction technique has a long history \cite{BM72}, and was employed in \cite{DN07},
\cite{DLN} and \cite{Dingconf} for obtaining linear codes with a few weights. Recently,
this trace construction of linear codes has attracted a lot of attention, and a
huge amount of linear codes with good parameters are obtained in
\cite{Ding151,DD15,HY15,HY152,Lifive,Tang5,XuCao,Zhoufour}.

It was claimed in \cite{Dingconf} that every linear code $\C$ over a finite field $\gf(q)$ has a trace
construction, i.e., $\C=\C_D$ for some defining set $D \in \gf(q^m)$, where $m$ is a positive integer.
However, to the best of the author's knowledge, no proof of this claim
is available in the literature. We will prove this statement shortly.

\section{A trace representation of linear codes over finite fields}\label{sec-tracerepresentation}

In this section, we will give a trace representation of all linear codes over finite fields. Specifically, we
will prove that any linear code $\C$ of length $n$ over $\gf(q)$ can be expressed as a trace code $\C_D$,
where $D$ is a subset of $\gf(q^m)$ for some positive integer $m$.

Let
\begin{eqnarray}\label{eqn-generatormatrixC}
\bD'=\left[
\begin{array}{llll}
d'_{1,1}  & d'_{1,2} & \cdots &  d'_{1,n} \\
d'_{2,1}  & d'_{2,2} & \cdots & d'_{2,n} \\
\vdots  & \vdots & \vdots & \vdots \\
d'_{k,1}  & d'_{k,2} & \cdots & d'_{k,n}
\end{array}
\right].
\end{eqnarray}
be any matrix whose row vectors span a linear code $\C$ of length $n$ over $\gf(q)$, where $k$ is equal to
or more than the dimension of $\C$.

Set $m=\max\{k, \lceil \log_q n \rceil \}$. By definition, we have $m \geq k$ and $n \leq q^m$.
Now put
$$
d'_i=\sum_{j=1}^k d'_{j,i} \alpha_j \mbox{ for all } 1 \leq i \leq n,
$$
where $\{ \alpha_1, \alpha_2, \cdots, \alpha_m\}$ is a basis of $\gf(q^m)$ over $\gf(q)$.

Define $D'=\{d'_1, d'_2, \cdots, d'_n\}$. It then follows from the
discussion in Section \ref{sec-generatormatrix} that the trace code
$\C_{D'}$ has generator matrix $\bD'$, and is thus equal to $\C$.

In general, we would have $m$ as small as possible when we wish to have a trace construction
of a linear code $\C$. To this end, we may have a generator matrix whose row vectors are linearly
independent over $\gf(q)$, i.e., the parameter $k$ above is equal to the dimension of $\C$.
It is now clear that a linear code over a finite field has many trace constructions (representations),
depending on the choice of a generator matrix.

\begin{example}
Let $q=2$ and $n=7$. Let $\C$ be the binary code with length $n=7$ and generator matrix
\begin{eqnarray}
G=\left[
\begin{array}{lllllll}
1 & 0 & 0 & 1 & 1 & 0 & 0 \\
0 & 1 & 0 & 0 & 1 & 1 & 0 \\
0 & 0 & 1 & 0 & 0 & 1 & 1
\end{array}
\right].
\end{eqnarray}
Then $\C$ has dimension $k=3$. Let $m=\max\{k, \lceil \log_q n \rceil \}=3$.
Let $\alpha$ be a generator of $\gf(2^3)^*$ with $\alpha^3+\alpha+1=0$.
Define $D=\{d_1, d_2, d_3, d_4, d_5, d_6, d_7\}$, where
$$
d_1=1, \  d_2=\alpha, \  d_3=\alpha^2, \ d_4=1, \ d_5=\alpha^3, \ d_6=\alpha^4, \ d_7=\alpha^2.
$$
Then the trace representation of $\C$ is the code $\C_D$ of (\ref{eqn-maincode}),
where $\tr$ is the trace function from $\gf(2^3)$ to $\gf(2)$.
\end{example}

\begin{example}
Let $q=2$ and $n=7$. Let $\C$ be the binary code with length $n=7$ and generator matrix
\begin{eqnarray}
G=\left[
\begin{array}{lllllll}
1 & 0 & 0 & 0 & 1 & 1 & 0 \\
0 & 1 & 0 & 0 & 0 & 1 & 1 \\
0 & 0 & 1 & 0 & 1 & 0 & 1 \\
0 & 0 & 0 & 1 & 0 & 1 & 1
\end{array}
\right].
\end{eqnarray}
Then $\C$ has dimension $k=4$. Let $m=\max\{k, \lceil \log_q n \rceil \}=4$.
Let $\alpha$ be a generator of $\gf(2^4)^*$ with $\alpha^4+\alpha+1=0$.
Define $D=\{d_1, d_2, d_3, d_4, d_5, d_6, d_7\}$, where
$$
d_1=1, \  d_2=\alpha, \  d_3=\alpha^2, \ d_4=\alpha^3, \ d_5=\alpha^8, \ d_6=\alpha^7, \ d_7=\alpha^{11}.
$$
Then the trace representation of $\C$ is the code $\C_D$ of (\ref{eqn-maincode}),
where $\tr$ is the trace function from $\gf(2^4)$ to $\gf(2)$.
\end{example}

The trace representation of this section gives naturally a trace construction of all cyclic codes
over finite fields without the condition that $\gcd(n, q)=1$. Recall that Theorem \ref{thm-tracecycliccodes}
requires that $\gcd(n,q)=1$.

\section{Another trace representation of all cyclic codes over finite fields}

In this section, we give a trace representation of all cyclic codes over finite fields.
This representation may be related to the $q$-polynomial approach to cyclic codes developed
in \cite{DingLing}.

Let $\C$ be a cyclic code of length $n$ over $\gf(q)$. A polynomial
$$
f(x)=\sum_{i=0}^{n-1} f_i x^i \in \gf(q)[x]
$$
generates $\C$ if and only if $\gcd(f(x), x^n-1)$ is equal to the generator polynomial of
$\C$. There are many such polynomials $f(x)$ that generates $\C$, e.g., the generator
polynomial and the generating idempotent of $\C$. Such a polynomial $f$ can be employed
to give a special trace representation of the code $\C$. The following theorem gives such a
representation.

\begin{theorem}\label{thm-allcyclic}
Let $f(x)=\sum_{i=0}^{n-1} f_i x^i \in \gf(q)[x]$ be any polynomial that generates a cyclic
code $\C$ of length $n$ over $\gf(q)$. Let $\alpha$ be a normal element of $\gf(q^n)$ over
$\gf(q)$. Define
$$
d=\sum_{j=1}^{n}  f_{n-j}  \alpha^{q^{j}}
$$
and
$$
D=\left\{d, d^q, d^{q^2}, \cdots, d^{q^{n-1}}\right\}.
$$
Then $\C=\C_D$, which is the trace code with defining set $D$, where the trace function is
from $\gf(q^n)$ to $\gf(q)$.
\end{theorem}

\begin{proof}
Since $f$ is a generator polynomial of the cyclic code $\C$, $\C$ has the following generator
matrix:
\begin{eqnarray*}
\left[
\begin{array}{llllllll}
f_0 & f_1 & f_2 & f_3 & \cdots & f_{n-3} & f_{n-2} & f_{n-1} \\
f_{n-1} & f_0 & f_1 & f_2 & \cdots & f_{n-3} & f_{n-1} & f_{n-2} \\
f_{n-2} & f_{n-1} & f_0 & f_1 & \cdots & f_{n-5} & f_{n-4} & f_{n-3} \\
\vdots & \vdots & \vdots & \vdots & \ddots & \vdots & \vdots & \vdots \\
f_{2} & f_{3} & f_4 & f_5 & \cdots & f_{n-1} & f_{0} & f_{1} \\
f_{1} & f_2 & f_3 & f_4 & \cdots & f_{n-2} & f_{n-1} & f_{0}
\end{array}
\right].
\end{eqnarray*}

Since $\alpha$ is a normal element of $\gf(q^n)$ over $\gf(q)$, $\{\alpha, \alpha^q, \cdots, \alpha^{q^{n-1}} \}$ is a normal basis of $\gf(q^n)$ over $\gf(q)$.
Let
$$
d_i=\sum_{j=0}^{n-1} f_{(j+i-1) \bmod{n}} \alpha^{q^j}.
$$
for all $i$ with $1 \leq i \leq n$.
Then $d_i=d^{q^{i-1}}$ for all $i$.
It then follows from the discussion in Section \ref{sec-tracerepresentation} that $\C=\C_D$.
\end{proof}

\begin{example}
Let $q=2$ and $n=7$. Let $\C$ be the binary cyclic code with length $n=7$ and generator polynomial
$f(x)=1+x+x^3$.
Then $\C$ has dimension $k=4$.
Let $\beta$ be a generator of $\gf(2^7)^*$ with $\beta^7+\beta+1=0$. Then $\alpha=\beta^{70}$
is a normal element of $\gf(2^7)$. Then
$$
d=\alpha^{2^6}+ \alpha^{2^5}+ \alpha^{2^3}=\beta^{14}.
$$
Define $D=\{d_1, d_2, d_3, d_4, d_5, d_6, d_7\}$, where
\begin{eqnarray*}
&& d_1=d^{q^0}=\beta^{14}, \\
&& d_2=d^{q^1}=\beta^{28}, \\
&& d_3=d^{q^2}=\beta^{56}, \\
&& d_4=d^{q^3}=\beta^{112}, \\
&& d_5=d^{q^4}=\beta^{97}, \\
&& d_6=d^{q^5}=\beta^{67}, \\
&& d_7=d^{q^6}=\beta^{7}.
\end{eqnarray*}
Then the trace representation of $\C$ is the code $\C_D$ of (\ref{eqn-maincode}),
where $\tr$ is the trace function from $\gf(2^7)$ to $\gf(2)$.
\end{example}

The trace representation of Theorem \ref{thm-tracecycliccodes} works under the condition that
$\gcd(q, n)=1$ only. Hence, it does not apply to all cyclic codes over finite fields. The advantage of
Theorem \ref{thm-allcyclic} is that it applies to all linear codes. But its disadvantage is that
the trace function is from a large extension field $\gf(q^n)$ to $\gf(q)$. The trace representation
of all linear codes over finite fields gives automatically a trace representation of all cyclic codes
over finite fields.

\section{Summary and concluding remarks}

The main contribution of this paper is the trace representation of all linear codes over finite
fields described in Section \ref{sec-tracerepresentation}. Consequently, any linear code $\C$ over
a finite field may be generated with a definition set $D$ via the trace construction.
Hence, all types of linear codes over finite fields, including all cyclic codes, have a trace
representation without having any restriction. This proves the claim in \cite{Dingconf} and confirms
that the defining-set method is indeed a fundamental construction of all linear codes over finite
fields.

The trace construction of linear codes $\C_D$ has the following advantages over the generator
matrix construction.
\begin{enumerate}
\item The description of the code $\C_D$ is simpler.
\item The determination of the weight distribution of $\C_D$ is much easier using the trace
         construction, as the weights of the codewords are expressed as character sums of the
         form (\ref{eqn-weight}).
\item It is possible to develop lower bounds on the minimum distance of linear codes when they
         are given as a trace code \cite{DY13}.
\end{enumerate}
It is observed that in almost all cases of the determination of the weight distribution of linear
codes, the trace construction has been employed.

A linear code over a finite field has many trace representations.  Special types of linear codes
may have special forms of trace representations. A trace representation of quasi-cyclic codes
was given in \cite{LingSole}.  A trace representation of quasi-negacyclic codes codes was
presented in \cite{LiFu}. It would be interesting to develop other trace
representations for special subclasses of linear codes.



%
%

\end{document}